\documentclass[12pt]{article}

\usepackage[margin=1in]{geometry}

\usepackage{amsmath, amssymb, amsthm}

\usepackage{enumitem}
\usepackage[hidelinks]{hyperref}
\usepackage[round]{natbib}
\setcitestyle{aysep={ }}
\usepackage{setspace}
\usepackage{newtxtext}
\usepackage{newtxmath}

\usepackage{authblk}

\usepackage{titlesec}
\titleformat{\section}{\normalfont\bfseries}{\thesection}{1em}{}
\titleformat{\subsection}{\normalfont\bfseries\itshape}{\thesubsection}{1em}{}
\titleformat{\subsubsection}{\normalfont\itshape}{\thesubsubsection}{1em}{}
\titleformat{\paragraph}[runin]{\normalfont\bfseries\itshape}{}{0pt}{}[.]
\titleformat{\subparagraph}[runin]{\normalfont\itshape}{}{0pt}{}[.]
\titlespacing*{\paragraph}{0pt}{0.5\baselineskip}{0.5em}
\titlespacing*{\subparagraph}{0pt}{0.5\baselineskip}{0.5em}

\usepackage{booktabs}

\usepackage{algorithm}
\usepackage{algpseudocode}

\newtheorem{theorem}{Theorem}
\newtheorem{proposition}{Proposition}

\newtheorem{assumption}{Assumption}

\newcommand{\Pzero}{P_0}
\newcommand{\Qbar}{\bar{Q}}
\newcommand{\Ybar}{\bar{Y}}
\newcommand{\aref}{\alpha_{\mathrm{ref}}}
\newcommand{\aadapt}{\alpha_{\mathrm{adapt}}}
\newcommand{\astab}{\alpha_{\mathrm{stab}}}
\newcommand{\op}{o_P}
\newcommand{\Op}{O_P}
\DeclareMathOperator*{\argmax}{arg\,max}

\title{\textbf{A CV-TMLE global test approach to improve power in rare disease clinical studies with multiple-component endpoints}}
\author[1]{Tianyue Zhou\thanks{Corresponding author. Email: tianyue\_zhou@berkeley.edu}}
\author[2]{Susan Gruber}
\author[3]{Hana Lee}
\author[3]{Wonyul Lee}
\author[3]{Lei Nie}
\author[1]{Mark van der Laan}
\affil[1]{Division of Biostatistics, School of Public Health, University of California, Berkeley, CA, USA}
\affil[2]{Targeted ML Solutions Inc., Cambridge, MA, USA}
\affil[3]{Center for Drug Evaluation and Research, U.S. Food and Drug Administration (FDA), Silver Spring, MD, USA}
\date{}

\begin{document}
\maketitle

\begin{abstract}
Rare disease trials face unique statistical challenges due to limited patient populations and heterogeneous clinical manifestations among patients. Multiple endpoints are often necessary to comprehensively capture treatment benefits. A global test is an approach for evaluating whether a treatment has any beneficial effect across multiple endpoints. We propose a new global test based on a weighted composite endpoint. The proposed global test employs shrinkage-based cross-validated targeted maximum likelihood estimation (CV-TMLE) to learn data-adaptive weights that maximize power while maintaining Type~I error control. Shrinkage can be tailored to incorporate existing domain knowledge, such as anticipated relative effect sizes. In simulation studies designed to reflect real rare disease trial settings, the proposed procedure demonstrated improved power over standard multiplicity adjustments and classical global tests (such as the O'Brien test), while maintaining nominal Type~I error, when effects are heterogeneous across endpoints. The proposed method simultaneously learns an optimal weighted composite outcome and provides an unbiased and efficient targeted maximum likelihood estimator (TMLE) for the average treatment effect (ATE) on that weighted outcome, with valid inference taking into account that the ATE is data dependent.
\end{abstract}

\noindent\textbf{Keywords:} CV-TMLE; global null hypothesis testing; multiple endpoints; rare disease clinical trials; data-adaptive target parameters

\medskip
\noindent\textbf{Main text word count:} 6472

\section{Introduction}

\subsection{Rare disease trials and the role of multiple endpoints}

Although more than 10{,}000 rare diseases have been identified, collectively affecting hundreds of millions of people worldwide, approved therapies exist for only a small fraction of these conditions \citep{tambuyzer2020}.
Assessing treatment efficacy in rare disease settings poses statistical challenges.
The small number of eligible patients severely constrains achievable sample sizes in randomized controlled trials (RCTs), making it difficult to demonstrate treatment benefit with adequate power \citep{cornu2013, korn2013}.
Compounding this, many rare diseases are poorly understood in terms of their natural history, which hampers the selection of appropriate study populations and clinically meaningful endpoints \citep{griggs2009, FDA2022rare}.
Rare diseases also frequently exhibit substantial genotypic and phenotypic heterogeneity within a single disorder: different patients may present with different clinical manifestations across multiple organ systems and respond to treatment in different disease domains \citep{augustine2013, chen2025challenges}.
This heterogeneity makes it particularly challenging to select a single primary endpoint that is sensitive to all clinically relevant manifestations of the disease.

For these reasons, investigators and regulators frequently rely on multiple endpoints to characterize treatment benefit.
Multiple endpoints may span distinct disease domains---such as motor function, respiratory capacity, biomarker levels, and patient-reported outcomes---allowing a broader assessment of whether a drug provides meaningful improvement \citep{huque2011, huque2013}.
This strategy is particularly valuable when it is uncertain a priori which disease domain will be most responsive to treatment within the feasible sample size.
Recent FDA guidance acknowledges that multiple endpoints can be an efficient strategy when no single endpoint adequately captures the treatment's overall clinical benefit \citep{FDA2022endpoints}, and further notes that rare disease trials should consider endpoint selection carefully given the heterogeneity of clinical presentations \citep{FDA2022rare}.

\subsection{The global null testing problem specific for rare diseases}

In rare disease trials, multiple endpoints often serve a different purpose than in common disease settings. Because limited sample sizes make it difficult to demonstrate benefit on any single endpoint, multiple endpoints may be used not to support specific labeling claims, but to answer a more fundamental question: \emph{is there any clinical benefit of this treatment?}
That is, before asking which specific endpoints are improved, one wants to know whether the overall body of evidence supports the conclusion that treatment provides benefit on at least one clinically meaningful outcome.
This is the \emph{global null hypothesis testing problem}: the global null states that the treatment has no beneficial effect on any of the endpoints under study, while the alternative is that the treatment improves at least one endpoint \citep{obrien1984, pocock1987}.
Rejection of the global null provides an omnibus assessment that the drug confers some clinical benefit, and is particularly valuable in rare disease trials where pooling information across endpoints can recover power that would otherwise be lost to the small sample size \citep{tang1993, shives2025}.

\subsection{Why standard global testing approaches can be underpowered}

Although the global null testing problem has been studied for decades, commonly used global tests can be underpowered in the rare disease setting for two key reasons.

The first limitation concerns \emph{fixed aggregation schemes}.
Classical global tests---such as the O'Brien rank-sum and generalized least-squares (GLS) tests \citep{obrien1984} and the Wei--Lachin multivariate tests \citep{wei1984}---combine endpoint-specific information into a single test statistic.
However, they typically rely on prespecified weighting schemes, most commonly equal weights or variance-based weights.
These fixed aggregation strategies perform well when treatment effects are relatively uniform across endpoints but can be substantially underpowered when the treatment effect pattern is heterogeneous.
When only a subset of endpoints carries a meaningful treatment signal, equal weighting dilutes the strong signal from informative endpoints with noise from uninformative ones.
In rare disease trials, where the treatment effect pattern across diverse disease domains is often unknown a priori, this rigidity can result in substantial power loss for the global question. Indeed, when the treatment effect pattern is heterogeneous, the power loss from fixed aggregation can be severe enough that a global test offers no advantage over traditional multiple comparison procedures that strongly control the family-wise error rate.

The second limitation concerns the \emph{lack of covariate adjustment}.
Many classical global tests are formulated as unadjusted group comparisons that do not incorporate baseline covariate information.
However, recent regulatory guidance strongly recommends prespecified adjustment for prognostic baseline covariates to improve statistical efficiency and power in randomized trials \citep{FDA2023covariates}.
Covariate adjustment is particularly important in rare disease settings: with small sample sizes, even moderate gains in efficiency from incorporating prognostic covariates can translate into meaningful increases in power \citep{FDA2023covariates}.
Ignoring available baseline information when samples are small means leaving statistical power on the table.

These two limitations---fixed aggregation that cannot adapt to heterogeneous effect patterns, and the absence of covariate adjustment that is especially costly at small sample sizes---jointly motivate the development of new global testing approaches that can (i) adaptively learn how to combine endpoint information and (ii) incorporate flexible covariate adjustment to maximize power for the global question.

\subsection{Contribution and broader relevance}

In this paper, we propose a general data-adaptive framework for global null
testing with multiple endpoints that directly addresses the two limitations
identified above. First, rather than relying on fixed aggregation schemes, we
prespecify a data-adaptive learning rule that selects the composite endpoint
weights that optimize power for the global question, allowing
the test to concentrate weight on the most informative endpoints when effects
are heterogeneous. Second, we employ CV-TMLE \citep{vanderlaan2011, zheng2011} to estimate
the treatment effect on the learned composite outcome, providing efficient and
flexible covariate adjustment together with valid inference for a data-adaptive
target parameter \citep{vanderlaan2015optimal, hubbard2016, vanderlaan2018}. The cross-validation structure separates
the data used to learn the composite weights from the data used to estimate
the resulting treatment effect, thereby preventing the overfitting and bias
that would arise from using the same observations for both parameter selection
and inference. A stabilization mechanism is further incorporated that ensures the procedure provides accurate control of the Type~I error rate.

Although motivated by rare disease RCTs, the proposed
framework is broadly applicable: the underlying TMLE methodology accommodates
observational designs, and the data-adaptive weight selection scales naturally
to settings with many correlated endpoints without incurring a penalty for multiple hypothesis testing.

The remainder of this paper is organized as follows.
Section~2 reviews related work on global testing, multiple testing, and
CV-TMLE.
Section~3 defines the statistical estimation problem, including the
weighted composite estimand and its efficient influence curve.
Section~4 presents the proposed stabilized CV-TMLE procedure.
Section~5 establishes theoretical properties: asymptotic linearity,
Type~I error control under the global null, and validity under alternatives.
Section~6 presents simulation studies evaluating the finite-sample
performance of the proposed procedure.
Section~7 concludes with a discussion.

\section{Related work}

\subsection{Global testing methods}

The problem of testing a global null hypothesis of no treatment difference across multiple endpoints has a long history in biostatistics.
O'Brien~\citep{obrien1984} proposed a family of procedures---most notably the nonparametric rank-sum test and the generalized least-squares (GLS) test---that combine endpoint-specific information into a single statistic, with power directed toward alternatives where treatment is uniformly better across all endpoints; however, both tests rely on fixed weighting schemes that can be suboptimal when effect sizes differ substantially across endpoints.
Wei and Lachin~\citep{wei1984} developed asymptotically distribution-free multivariate tests that accommodate incomplete observations and avoid strong distributional assumptions, though their performance can also degrade under heterogeneous treatment effect patterns.
Subsequent work refined these methods: Huang et al.~\citep{huang2005} corrected the asymptotic variance in O'Brien's rank-sum test to improve Type~I error control under heteroscedasticity; Logan and Tamhane~\citep{logan2005} proposed improved small-sample $t$-distribution approximations for the OLS and GLS statistics; and Tang et al.~\citep{tang1993} provided design recommendations for global tests with multiple endpoints.
Pocock et al.~\citep{pocock1987} presented a general framework applicable to any set of asymptotically normal endpoint-specific test statistics.
A separate line of work constructs a global statistic from endpoint-wise test statistics and calibrates it via the permutation or randomization distribution, naturally accounting for endpoint dependence without parametric assumptions \citep{westfall1993, pesarin2010}; however, permutation calibration alone does not address how to optimally combine endpoints under heterogeneous alternatives.
Across all of these approaches, a common limitation is the reliance on fixed aggregation schemes that cannot adapt to unknown treatment effect patterns---the gap our proposed method aims to fill.

\subsection{Multiple testing procedures}

The clinical trials multiplicity literature offers a complementary perspective to global testing by providing tools for controlling the familywise error rate (FWER) when making endpoint-specific claims.
The simplest approach is the Bonferroni correction, which divides the significance level equally among all endpoints, guaranteeing FWER control under arbitrary dependence but often at the cost of substantial conservatism, especially when endpoints are positively correlated \citep{dmitrienko2013}.
Holm's step-down procedure \citep{holm1979} uniformly dominates the Bonferroni correction by testing ordered $p$-values against progressively less stringent thresholds, while Hochberg's step-up procedure \citep{hochberg1988} is uniformly more powerful still when its validity conditions---such as positive dependence among test statistics---are satisfied \citep{sarkar1997, huque2016}.
When endpoints have a prespecified clinical hierarchy, gatekeeping procedures and the graphical approach of Bretz et al.~\citep{bretz2009, bretz2011} provide flexible frameworks for allocating the Type~I error budget across endpoint families \citep{dmitrienko2013, huque2013}.
These multiple testing procedures are primarily designed for multiplicity control when multiple tests are performed to support approval and/or labeling claims.
While one can use a multiple testing procedure as a global test by rejecting the global null whenever at least one adjusted $p$-value is significant, this ``reject-if-any'' construction does not pool evidence across endpoints and can be conservative when treatment effects are moderate and diffuse.

\subsection{CV-TMLE and data-adaptive target parameters}

In contrast to the fixed-weight global tests and multiple testing procedures described above, our approach builds on the semiparametric estimation literature---specifically, targeted maximum likelihood estimation (TMLE) and the theory of data-adaptive target parameters.
TMLE is a general framework for efficient estimation of statistical estimands: it first obtains an initial estimate of the outcome regression using potentially flexible machine learning methods, then updates it via a targeting step that solves the efficient influence function equation, yielding a doubly robust, locally efficient estimator with valid inference \citep{vanderlaan2011}.
Cross-validated TMLE (CV-TMLE) extends this framework by incorporating $V$-fold sample splitting, estimating nuisance parameters on training folds and performing the targeting step and parameter evaluation on held-out validation folds, which improves robustness with adaptive nuisance estimators and facilitates valid inference under weaker regularity conditions \citep{vanderlaan2011, zheng2011}.
Crucially, \citet{vanderlaan2015optimal} introduced a general CV-TMLE procedure for \emph{data-adaptive target parameters}---settings in which the estimand itself is chosen based on the observed data. In this framework, initial nuisance estimates are fit on training folds together with the target parameter selection, while only the TMLE targeting step is performed on validation folds; sample splitting separates the parameter selection step from estimation and inference, thereby preventing the optimistic bias that would arise from using the same data for both \citep[see also][Section~9.6]{vanderlaan2018}. This is the framework we adopt.
In our setting, the composite endpoint weights define the target parameter and are learned from the data; the CV-TMLE framework naturally accommodates this by learning the weights and fitting nuisance parameters on training folds and performing the targeting step and parameter evaluation on validation folds, producing a final cross-validated estimate with valid inference \citep{zheng2011, vanderlaan2015optimal, hubbard2016, vanderlaan2018}.

\section{Defining the statistical estimation problem}\label{sec:setup}

\subsection{Observed data and notation}\label{sec:data}

Consider a two-arm RCT enrolling $n$ subjects.
For each subject $i = 1, \ldots, n$, we observe
\begin{equation}\label{eq:data}
  O_i = (W_i, A_i, Y_i), \quad i = 1, \ldots, n,
\end{equation}
where $W_i \in \mathcal{W} \subseteq \mathbb{R}^d$ is a vector of baseline
covariates, $A_i \in \{0, 1\}$ is the treatment assignment indicator, and
$Y_i = (Y_i^{(1)}, \ldots, Y_i^{(K)})^\top \in \mathbb{R}^K$
is the vector of $K \geq 2$ endpoints.
The observations are independent and identically distributed according to
$P_0 \in \mathcal{M}$, where $\mathcal{M}$ is a nonparametric model.
Each outcome is oriented so that larger values indicate a more favorable
response.
The propensity score $g_0(1 \mid W) = P_0(A = 1 \mid W)$ is known by design;
in a trial with a 1:1 randomization ratio, $g_0(1 \mid W) = 0.5$ for all $W$.

\subsection{Causal estimand, identification, and statistical estimand}\label{sec:estimands}

For each endpoint $k \in \{1, \ldots, K\}$, let $Y^{(k)}(a)$ denote the
potential outcome under treatment $a \in \{0, 1\}$.
The causal average treatment effect (ATE) for endpoint $k$ is
\begin{equation}\label{eq:ate}
  \mathrm{ATE}_k \;=\; \mathbb{E}\bigl[Y^{(k)}(1) - Y^{(k)}(0)\bigr].
\end{equation}
Under the following identification assumptions, the causal ATE is linked
to a statistical estimand.
The \emph{randomization} assumption
($A \perp\!\!\!\perp Y(0), Y(1) \mid W$) and the \emph{positivity}
assumption ($P(A = a \mid W) > 0$ for $a \in \{0,1\}$ and all $W$ in
the support) both hold by the RCT design.
We additionally assume \emph{consistency}: the observed outcome under
treatment $a$ equals the potential outcome, i.e., if $A = a$, then
$Y = Y(a)$.
For simplicity, we further assume no loss to follow-up or missingness,
so that the full outcome vector $Y$ is observed for every subject;
note that this assumption is \emph{not} guaranteed by the RCT design
and must be assessed in practice.
Under these conditions, the causal ATE is identified by the statistical
estimand
\begin{equation}\label{eq:stat-estimand}
  \psi_k(P_0) \;=\; E_{P_0}\bigl[
    \Qbar_k(1, W) - \Qbar_k(0, W)
  \bigr],
\end{equation}
where $\Qbar_k(A, W) = E_{P_0}(Y^{(k)} \mid A, W)$ is the outcome
regression for endpoint $k$.

\subsection{Global null hypothesis}\label{sec:hypotheses}

The primary inferential question is whether the treatment provides benefit
on at least one endpoint. We formalize this as a one-sided global null test:
\begin{equation}\label{eq:global-null}
  H_0 : \psi_k(P_0) \leq 0 \;\;\text{for all } k = 1, \ldots, K
  \qquad \text{versus} \qquad
  H_1 : \psi_k(P_0) > 0 \;\;\text{for some } k.
\end{equation}
Rejection of $H_0$ supports the conclusion that treatment improves at
least one endpoint, but does not identify which endpoint(s) are improved.

\subsection{Weighted composite estimand}\label{sec:composite}

We construct a global test by combining endpoint-specific effects into a
single summary. For a weight vector
$\alpha = (\alpha_1, \ldots, \alpha_K)^\top$ on the simplex
$\Delta^{K-1} = \{ \alpha \in \mathbb{R}^K : \alpha_k \geq 0,\;
\sum_k \alpha_k = 1 \}$,
define the weighted composite outcome
$\Ybar^{(\alpha)} = \sum_{k} \alpha_k Y^{(k)}$ and the composite ATE
\begin{equation}\label{eq:composite-ate}
  \psi_\alpha(P_0) \;=\; \sum_{k=1}^K \alpha_k\, \psi_k(P_0)
  \;=\; E_{P_0}\bigl[\Qbar_\alpha(1, W) - \Qbar_\alpha(0, W)\bigr],
\end{equation}
where $\Qbar_\alpha(A, W) = \sum_k \alpha_k \Qbar_k(A, W)$.
Since each $\alpha_k \geq 0$, the global null $H_0$ implies
$\psi_\alpha(P_0) \leq 0$ for every $\alpha \in \Delta^{K-1}$.
Therefore, when $\alpha$ is fixed, a valid one-sided test of $\psi_\alpha(P_0) > 0$ constitutes a valid global null test.
When $\alpha$ is selected adaptively, additional care is needed to preserve
Type~I error control; this is addressed in Section~\ref{sec:method}.
Note that the set of endpoints can be augmented with non-linear combinations of domain-specific endpoints, such as the maximum improvement over all domains or the number of domains that showed improvement. This can potentially increase power without affecting Type~I error control.

\subsection{Efficient influence curve and asymptotic variance}\label{sec:eic}

For each endpoint-specific estimand $\psi_k(P_0)$, the efficient influence
curve (EIC) in the nonparametric model is
\begin{equation}\label{eq:eic-k}
  D^*_{\psi_k, P}(O) \;=\;
  \frac{2A - 1}{g_0(A \mid W)}
  \bigl(Y^{(k)} - \Qbar_k(A, W)\bigr)
  \;+\; \Qbar_k(1, W) - \Qbar_k(0, W) - \psi_k(P).
\end{equation}
By linearity, the EIC for the composite estimand is
$D^*_{\psi_\alpha, P}(O) = \sum_{k} \alpha_k \, D^*_{\psi_k, P}(O)$,
or equivalently,
\begin{equation}\label{eq:eic-alpha}
  D^*_{\psi_\alpha, P}(O) \;=\;
  \frac{2A - 1}{g_0(A \mid W)}
  \bigl(\Ybar^{(\alpha)} - \Qbar_\alpha(A, W)\bigr)
  \;+\; \Qbar_\alpha(1, W) - \Qbar_\alpha(0, W) - \psi_\alpha(P).
\end{equation}
The asymptotic variance of an efficient estimator of $\psi_\alpha(P_0)$ is
\begin{equation}\label{eq:var-alpha}
  \sigma^2_\alpha(P_0) \;=\; \mathrm{Var}_{P_0}\!\bigl(
    D^*_{\psi_\alpha, P_0}(O)
  \bigr) \;=\;
  \sum_{k_1, k_2}
    \alpha_{k_1} \alpha_{k_2}\, \rho_{P_0}(k_1, k_2),
\end{equation}
where $\rho_{P_0}(k_1, k_2) =
E_{P_0}[D^*_{\psi_{k_1}, P_0}(O)\, D^*_{\psi_{k_2}, P_0}(O)]$
is the covariance between endpoint-specific influence curves.
The variance $\sigma^2_\alpha(P_0)$ is thus a quadratic form in $\alpha$,
with the matrix $\boldsymbol{\rho}_{P_0} =
[\rho_{P_0}(k_1, k_2)]_{k_1, k_2}$ capturing the full dependence structure.

\section{Proposed testing framework}\label{sec:method}

We now present the stabilized CV-TMLE procedure for testing the global null hypothesis
\eqref{eq:global-null} introduced in Section~\ref{sec:hypotheses}.
The procedure has three conceptual components: (i) a criterion for
selecting optimal composite weights that maximize the signal-to-noise
ratio of the resulting test,
(ii) a cross-validation structure
that separates weight learning from estimation and inference, and
(iii) a stabilization mechanism that ensures Type~I error control
when the weight optimization is degenerate under the null.

\subsection{Oracle weight selection}\label{sec:oracle}

The power of a one-sided Wald test of
$H_0(\alpha)\!: \psi_\alpha(\Pzero) \leq 0$ at sample size $n$ is an
increasing function of the signal-to-noise ratio
$\theta_\alpha(\Pzero) = \psi_\alpha(\Pzero)/\sigma_\alpha(\Pzero)$,
where $\sigma_\alpha(\Pzero) = [\sigma^2_\alpha(\Pzero)]^{1/2}$ is
the asymptotic standard deviation defined in \eqref{eq:var-alpha}.
Among the family of valid global tests indexed by
$\alpha \in \Delta^{K-1}$, the most powerful test is obtained at the
\emph{oracle weight}
\begin{equation}\label{eq:oracle}
  \alpha_0 \;=\; \argmax_{\alpha \in \Delta^{K-1}}\;
  \theta_\alpha(\Pzero)
  \;=\; \argmax_{\alpha \in \Delta^{K-1}}\;
  \frac{\psi_\alpha(\Pzero)}{\sigma_\alpha(\Pzero)},
\end{equation}
The oracle depends on both the unknown treatment effect pattern and
the endpoint covariance structure through $\Pzero$, and thus cannot
be specified in advance.

Classical global tests correspond to particular fixed choices of
$\alpha$.
The O'Brien OLS test \citep{obrien1984} sets $\alpha_k = 1/K$, which
is optimal when treatment effects are homogeneous across endpoints.
The general framework of \citet{pocock1987} accommodates any set of
asymptotically normal endpoint-specific test statistics combined via
fixed weights.
All of these procedures use prespecified weights that cannot adapt to
an unknown effect pattern; the oracle $\alpha_0$ generalizes them by
letting the data-generating distribution determine the optimal
combination.

\subsection{Data-adaptive weight estimation}\label{sec:weight-est}

We propose to estimate $\alpha_0$ from data via
\begin{equation}\label{eq:alpha-hat}
  \hat{\alpha}(P_n) \;=\; \argmax_{\alpha \in \Delta^{K-1}}\;
  \frac{\hat{\psi}^*_{\alpha,n}}{\hat{\sigma}_{\alpha,n}},
\end{equation}
where $P_n$ denotes the empirical distribution of the observed data
$(O_1, \ldots, O_n)$,
$\hat{\psi}^*_{k,n}$ is the TMLE of the endpoint-specific treatment
effect $\psi_k(\Pzero)$ computed from the full sample,
$\hat{\psi}^*_{\alpha,n} = \sum_k \alpha_k \hat{\psi}^*_{k,n}$
is the TMLE of $\psi_\alpha(\Pzero)$, and $\hat{\sigma}_{\alpha,n}$
is the corresponding plug-in standard deviation estimate.
Because the estimand $\psi_{\hat{\alpha}}(\Pzero)$ now depends on
the observed data through the learned weights, this is a
\emph{data-adaptive target parameter} in the sense of
\citet{hubbard2016}.

The fundamental challenge is that using the same data to both select
$\alpha$ and estimate $\psi_\alpha$ creates optimistic bias: the
maximization in \eqref{eq:alpha-hat} tends to select weights that
capitalize on sampling noise, inflating the test statistic.
This is precisely the setting that motivates sample splitting and
cross-validated estimation \citep{hubbard2016}.

\subsection{CV-TMLE procedure}\label{sec:cvtmle}

We employ the CV-TMLE framework \citep{vanderlaan2011, zheng2011}
as the natural vehicle for estimation and inference with data-adaptive
target parameters.
Randomly partition the $n$ observations into $V$ approximately
equal-sized folds $\mathcal{V}_1, \ldots, \mathcal{V}_V$.
For each fold $v$, let $P_{n,v}$ and $P^1_{n,v}$ denote the empirical
measures of the training sample (all folds except $v$) and the
validation sample (fold $v$), respectively.
The procedure operates in three steps.

\paragraph{Step 1: Weight optimization on training data}
For each fold $v$, compute endpoint-specific TMLEs
$\hat{\psi}^*_{k,\mathrm{train}}$ $(k = 1, \ldots, K)$ on the training
sample, extract the estimated influence curves
$\hat{D}^*_{\psi_k}(O_i)$ for $i$ in the training set, and form the
estimated covariance matrix
$\hat{\rho}_n(k_1, k_2) = \widehat{\mathrm{Cov}}(
\hat{D}^*_{\psi_{k_1}}, \hat{D}^*_{\psi_{k_2}})$.
Solve the optimization problem
\begin{equation}\label{eq:grid}
  \aadapt^{(v)} \;=\;
  \argmax_{\alpha \in \Delta^{K-1}}
  \frac{\sum_k \alpha_k \hat{\psi}^*_{k,\mathrm{train}}}
       {\bigl(\sum_{k_1,k_2}
         \alpha_{k_1}\alpha_{k_2}\hat{\rho}_n(k_1,k_2)\bigr)^{1/2}},
\end{equation}
using, for example, grid search over the simplex, gradient-based
optimization, or any other suitable numerical method;
record both the optimal adaptive weights $\aadapt^{(v)}$ and the
corresponding test statistic
$T^*_v = \sqrt{n_{\mathrm{train}}} \max_{\alpha \in \Delta^{K-1}}
\hat{\psi}^*_{\alpha,\mathrm{train}} / \hat{\sigma}_{\alpha,\mathrm{train}}$.

\paragraph{Step 2: Estimation on validation data}
Using the learned weights $\aadapt^{(v)}$, form the
composite outcome ${\Ybar}^{(\aadapt^{(v)})}_{i} = \sum_k \alpha_{\mathrm{adapt},k}^{(v)}
Y^{(k)}_i$ for each validation-fold observation.
Form the initial composite outcome regression as the corresponding linear combination of the endpoint-specific models from Step~1,
$\Qbar_{\aadapt^{(v)}}(A, W) = \sum_k \alpha_{\mathrm{adapt},k}^{(v)} \, \Qbar_{k,\mathrm{train}}(A, W)$,
and apply a targeting step on
the validation sample to obtain the fold-specific estimate
$\hat{\psi}^*_v$ and cross-fitted influence curve values.
The key structural feature is that both the weight selection and
the initial outcome regression are determined by the training data,
while only the TMLE targeting step and parameter evaluation use
the independent validation data, precisely the cross-validated TMLE
structure for data-adaptive target parameters
\citep{vanderlaan2015optimal, vanderlaan2018}.
In practice, a pooled targeting step, in which a single fluctuation parameter $\epsilon$ is estimated across all validation folds simultaneously, can be used in place of fold-specific targeting; this can stabilize the TMLE update, which is especially beneficial in small samples.

\paragraph{Step 3: Pooled inference}
The CV-TMLE point estimate is the average across folds,
\begin{equation}\label{eq:cvtmle-est}
  \hat{\psi}^*_{\mathrm{CV}}
  \;=\; \frac{1}{V}\sum_{v=1}^V \hat{\psi}^*_v,
\end{equation}
with cross-validated variance estimate
$\hat{\sigma}^2_{\mathrm{CV}} = n^{-1}\sum_{i=1}^n
\bigl(\hat{D}^*_i - \bar{D}^*_n\bigr)^2$,
where $\bar{D}^*_n = n^{-1}\sum_{i=1}^n \hat{D}^*_i$ and $\hat{D}^*_i$ is the influence curve value
for subject $i$ from the fold in which it served as a validation
observation.
The test statistic is
$T_{\mathrm{CV}} = \sqrt{n}\,\hat{\psi}^*_{\mathrm{CV}} /
\hat{\sigma}_{\mathrm{CV}}$, referred to a $t$-distribution with
$\nu$ degrees of freedom following the small-sample correction of
\citet{logan2005}.

\subsection{The non-uniqueness problem under the null}\label{sec:nonunique}

A fundamental challenge arises when applying the procedure above to
global null hypothesis testing.
Under $H_0\!: \psi_k(\Pzero) \leq 0$ for all $k$, we have
$\psi_\alpha(\Pzero) = \sum_k \alpha_k \psi_k(\Pzero) \leq 0$
for every $\alpha \in \Delta^{K-1}$ (since $\alpha_k \geq 0$),
so the signal-to-noise ratio $\theta_\alpha(\Pzero) \leq 0$
for all $\alpha$.
At the boundary of $H_0$ where $\psi_k(\Pzero) = 0$ for all $k$,
the numerator vanishes identically and the oracle $\alpha_0$ in
\eqref{eq:oracle} is non-unique---the entire simplex is a
set of maximizers.
More generally, whenever at least two $\psi_k(\Pzero) = 0$, the
maximizer set is not a singleton, since any convex combination
supported on the zero coordinates achieves $\theta_\alpha(\Pzero) = 0$.
This is qualitatively different from the standard data-adaptive
target parameter setting \citep{vanderlaan2015optimal, hubbard2016},
where the parameter-generating algorithm is assumed to converge to a
well-defined limit.

The weight estimator $\aadapt^{(v)}$ does
not converge to a fixed point but instead drifts on the simplex,
driven by sampling noise in the training-fold estimates.
Although each fold's CV-TMLE is still a valid estimator of a
(fold-specific) composite treatment effect that is at most zero under
the null, the maximization over $\alpha$ distorts the sampling
distribution of the test statistic: it tends to select weights that
amplify chance positive deviations, consequently inflating Type~I error.

\subsection{Stabilized CV-TMLE}\label{sec:stabilized}

We resolve the non-uniqueness problem by introducing a shrinkage step
that blends the data-adaptive weights with prespecified reference
weights.
The shrinkage is governed by the strength of the evidence for a
treatment effect, as measured by the training-fold $p$-value.

Specifically, after computing the adaptive weights
$\aadapt^{(v)}$ and the maximized training-fold statistic $T^*_v$
in Step~1, we compute the $p$-value
$p_{n,v} = P_{H_0}(T^* \ge T^*_v)$,
where the probability is taken under the null distribution of the
supremum statistic $T^* = \sup_{\alpha \in \Delta^{K-1}} T_\alpha$.
Asymptotically, this distribution is that of
$\sup_\alpha \alpha^\top Z \big/ \sqrt{\alpha^\top \Sigma \alpha}$
with $Z \sim N(0, \Sigma)$, where $\Sigma$ is the true covariance matrix of the influence curves; in practice, we plug in the consistent estimate $\hat\Sigma$ for Monte Carlo evaluation.

The stabilized weights are then
\begin{equation}\label{eq:shrinkage}
  \astab^{(v)} \;=\;
  \bigl(1 - \min(1,\,\tilde{C}\, p_{n,v})\bigr)\,\aadapt^{(v)}
  \;+\; \min(1,\,\tilde{C}\, p_{n,v})\,\aref,
\end{equation}
where $\aref = (1/K, \ldots, 1/K)$ are the reference weights and
$\tilde{C} = C\log(n)$ for a constant $C > 0$.
These stabilized weights $\astab^{(v)}$ replace $\aadapt^{(v)}$ in
Step~2 of the procedure; all subsequent steps are unchanged.

When evidence for a treatment effect is weak ($p_{n,v}$ is not small),
the product $\tilde{C}\, p_{n,v}$ exceeds~1 and the weights collapse
toward $\aref$; the test then behaves like the O'Brien test, which controls Type~I error at the nominal level.
Conversely, when evidence is strong ($p_{n,v} \to 0$ faster than
$1/\log(n)$), the shrinkage coefficient vanishes and the data-adaptive
weights are recovered, preserving the procedure's ability to concentrate
weight on the most informative endpoints.

The complete procedure is summarized in Algorithm~\ref{alg:stab}.

\begin{algorithm}[t]
\caption{Stabilized cross-validated TMLE}\label{alg:stab}
\begin{algorithmic}[1]
\Require Data $(W_i, A_i, Y_i)_{i=1}^n$; folds $V$; reference weights
  $\aref$; constant $C > 0$
\State Randomly partition $\{1, \ldots, n\}$ into $V$ folds
\For{fold $v = 1, \ldots, V$}
  \State \textbf{Step 1a (Optimize):} On training data, compute
    endpoint-specific TMLEs $\hat{\psi}^*_{k,\mathrm{train}}$ and
    covariance $\hat{\rho}_n$; solve \eqref{eq:grid} for
    $\aadapt^{(v)}$; record $T^*_v$
  \State \textbf{Step 1b (Stabilize):} Compute $p_{n,v} = P_{H_0}(T^* \ge T^*_v)$
    via Monte Carlo from the supremum null distribution;
    set $\astab^{(v)}$ via \eqref{eq:shrinkage}
  \State \textbf{Step 2 (Validate):} Form composite outcome
    $\Ybar^{(\astab^{(v)})}_i$ on validation fold; form
    $\Qbar_{\astab^{(v)}} = \sum_k \alpha_{\mathrm{stab},k}^{(v)} \Qbar_{k,\mathrm{train}}$; apply TMLE targeting on
    validation data; store $\hat{\psi}^*_v$ and influence curve values
\EndFor
\State \textbf{Step 3 (Pool):} Compute
  $\hat{\psi}^*_{\mathrm{CV}} = V^{-1}\sum_v \hat{\psi}^*_v$,\;
  $\hat{\sigma}^2_{\mathrm{CV}} = n^{-1}\sum_i \bigl(\hat{D}^{*}_i - \bar{D}^*_n\bigr)^2$,\;
  $T_{\mathrm{CV}} = \sqrt{n}\,\hat{\psi}^*_{\mathrm{CV}}/
  \hat{\sigma}_{\mathrm{CV}}$
\State \textbf{Decision:} Reject $H_0$ at significance level $\gamma$ if
  $T_{\mathrm{CV}} > t_{1-\gamma}(\nu)$
\end{algorithmic}
\end{algorithm}

\subsection{Choice of tuning parameters}\label{sec:tuning}

The procedure involves three user-specified quantities: the reference
weights $\aref$, the shrinkage constant $C$, and the number of
cross-validation folds $V$.

The default reference weights $\aref = (1/K, \ldots, 1/K)$ correspond
to the O'Brien test and represent a natural, interpretable
baseline.
When domain knowledge suggests a particular effect pattern---for
example, that a primary endpoint is more likely to show benefit---the
reference weights can be adjusted accordingly.
This flexibility allows prior clinical information to be incorporated
without sacrificing the ability to adapt if the data suggest a
different pattern.

The constant $C$ controls the rate of shrinkage toward the reference
weights.
The choice of $C$ should be guided by simulation studies conducted
a priori, calibrated to balance Type~I error control against power
for the anticipated sample size and endpoint configuration.

The number of cross-validation folds $V$ may be chosen depending on
the sample size and the bias--variance trade-off inherent in the
sample-splitting scheme.
In our setting, we found that $V = 10$ folds performed well,
consistent with standard practice in cross-validated TMLE
\citep{zheng2011}.

For the degrees of freedom $\nu$ in the final test statistic
$T_{\mathrm{CV}}$, we adopt the small-sample correction of
\citet{logan2005} for the O'Brien-type statistic.
The training-fold $p$-value $p_{n,v}$ is computed by Monte Carlo
simulation from the null distribution of the supremum statistic
(see Section~\ref{sec:stabilized}), using $B = 5{,}000$ draws.

\section{Theoretical properties}\label{sec:theory}

We now establish the key theoretical properties of the stabilized
CV-TMLE procedure.
The analysis draws on the general CV-TMLE theory for data-adaptive target
parameters \citep{vanderlaan2015optimal, vanderlaan2018}
and the CV-TMLE framework
\citep{vanderlaan2011, zheng2011}, with a novel argument needed to
handle the non-uniqueness of the oracle weight under the global null.
We first state the regularity conditions and derive asymptotic
linearity of the CV-TMLE estimator
(Section~\ref{sec:assumptions}), then prove Type~I error control
under the global null (Section~\ref{sec:null-theory}), and finally
establish validity and power under alternatives
(Section~\ref{sec:alt-theory}).

\subsection{Regularity conditions and asymptotic linearity}\label{sec:assumptions}

The main result of this subsection is asymptotic linearity of the
cross-validated TMLE estimator $\hat{\psi}^*_{\mathrm{CV}}$ for the
data-adaptive composite treatment effect.
We establish this by verifying, in the present RCT setting, the
conditions of a general result for CV-TMLE with data-adaptive target
parameters due to \citet{vanderlaan2015optimal} (see also
\citealp{vanderlaan2018}, Section~9.6 and Theorem~A.1),
which we restate here in our notation.
The key feature of this result, and the reason it applies to our
Algorithm~1, is that the initial nuisance estimates are fit on
the training sample together with the data-adaptive target
parameter selection, while only the TMLE targeting step is carried out on
the validation sample.

\begin{theorem}[{\citet{vanderlaan2015optimal};
\citealp{vanderlaan2018}, Theorem~A.1}]\label{thm:cvtmle-datp}
Consider a $V$-fold cross-validation scheme with training empirical
measures $P_{n,v}$ and validation empirical measures $P^1_{n,v}$.
For each fold $v$, suppose an algorithm applied to the training
sample $P_{n,v}$:
\begin{enumerate}[label=(\alph*)]
  \item selects a target parameter
  $\Psi_v(\Pzero) = \psi_{\alpha_{n,v}}(\Pzero)$; and
  \item fits initial nuisance estimators
  $\Qbar_{\alpha_{n,v},n,v}$ and $g_{n,v}$ on $P_{n,v}$.
\end{enumerate}
A TMLE targeting step is then applied on the validation sample
$P^1_{n,v}$, producing the updated
$\Qbar^*_{\alpha_{n,v},n,v}$ and the fold-specific plug-in estimate
$\hat{\psi}^*_v = \psi_{\alpha_{n,v}}(\Qbar^*_{\alpha_{n,v},n,v})$.
Let
$D^*_v \equiv D^*_{\psi_{\alpha_{n,v}},\,
\Qbar^*_{\alpha_{n,v},n,v},\, g_{n,v}}$
denote the estimated efficient influence curve on fold~$v$, and let
$R_{\alpha}(\Qbar, g, \Qbar_0, g_0)$ denote the exact second-order
remainder for $\psi_\alpha$, satisfying
$R_\alpha(\Qbar, g, \Qbar_0, g_0)
= \psi_\alpha(\Qbar) - \psi_\alpha(\Pzero)
+ \Pzero D^*_{\psi_\alpha, \Qbar, g}$.
Assume:
\begin{enumerate}[label=(\roman*)]
  \item \emph{Cross-validated score equation:}
  $V^{-1}\sum_{v=1}^V P^1_{n,v}\, D^*_v = \op(n^{-1/2})$.
  \item \emph{Bounded influence curves:}
  $\|D^*_v\|_\infty \leq M < \infty$ with probability tending to~$1$.
  \item \emph{Influence curve convergence:}
  There exists a fixed limit $D^*(\Pzero)$ such that
  $\Pzero\{D^*_v - D^*(\Pzero)\}^2 \to 0$ in probability.
  \item \emph{Negligible remainder:}
  $V^{-1}\sum_{v=1}^V R_{\alpha_{n,v}}(
  \Qbar^*_{\alpha_{n,v},n,v}, g_{n,v}, \Qbar_0, g_0) = \op(n^{-1/2})$.
\end{enumerate}
Then the cross-validated estimator
$\hat{\psi}^*_{\mathrm{CV}} = V^{-1}\sum_{v=1}^V \hat{\psi}^*_v$
satisfies
\[
  \sqrt{n}\bigl(
    \hat{\psi}^*_{\mathrm{CV}}
    - V^{-1}\textstyle\sum_{v} \psi_{\alpha_{n,v}}(\Pzero)
  \bigr)
  \;\xrightarrow{d}\;
  N\bigl(0,\;\Pzero\{D^*(\Pzero)\}^2\bigr).
\]
\end{theorem}

\noindent
The sample-splitting structure eliminates the need for a Donsker
class condition on $\{D^*_v\}$, permitting arbitrarily complex
data-adaptive algorithms
\citep{vanderlaan2015optimal, vanderlaan2018}.
We now introduce the regularity conditions under which
Theorem~\ref{thm:cvtmle-datp} applies to our problem,
and verify each of its four conditions in the present RCT setting.

\begin{assumption}[Randomized treatment assignment]\label{as:rct}
The data arise from an RCT in which the
treatment assignment mechanism $g_0(a \mid W) = \Pzero(A = a \mid W)$ is known
by design and satisfies
$\delta \leq g_0(1 \mid W) \leq 1 - \delta$ for some $\delta \in (0, 0.5)$ and all $W \in \mathcal{W}$.
\end{assumption}

\begin{assumption}[Bounded outcomes]\label{as:bounded}
The endpoint vector $Y = (Y^{(1)}, \ldots, Y^{(K)})$ lies in a
bounded subset of $\mathbb{R}^K$ almost surely.
\end{assumption}

\begin{assumption}[Consistent nuisance estimation]\label{as:nuisance}
For each fold $v$, the training-sample outcome regression estimator
$\Qbar_{k,n,v}$ converges in $L^2(\Pzero)$ to the true conditional mean
$\Qbar_{0,k}$, uniformly over $k$ and $v$.
\end{assumption}

We now verify the four conditions of Theorem~\ref{thm:cvtmle-datp}.
In our setting, the algorithm applied to training fold~$P_{n,v}$
selects weights $\alpha_{n,v} \in \Delta^{K-1}$, defining the
fold-specific target parameter
$\psi_{\alpha_{n,v}}(\Pzero)$, and fits the initial outcome
regression $\Qbar_{\alpha_{n,v},n,v}$ on the same training sample.
Because the treatment mechanism $g_0$ is known
(Assumption~\ref{as:rct}), we use it directly in the targeting step
rather than estimating it, so $g_{n,v} = g_0$.

\paragraph{Condition~(i): Cross-validated score equation}
If fold-specific targeting is used, the TMLE update solves the
efficient influence curve equation exactly on each validation fold:
\[
  P^1_{n,v}\, D^*_{\psi_{\alpha_{n,v}},\, \Qbar^*_{\alpha_{n,v},n,v},\, g_0} = 0
  \quad \text{for each } v.
\]
If pooled targeting is used (Section~\ref{sec:cvtmle}, Step~2),
the cross-validated average
\[
  \frac{1}{V}\sum_v P^1_{n,v}\, D^*_{\psi_{\alpha_{n,v}},\, \Qbar^*_{\alpha_{n,v},n,v},\, g_0} = 0
\]
holds exactly.

\paragraph{Condition~(ii): Bounded influence curves}
By Assumption~\ref{as:bounded} and the positivity
bound in Assumption~\ref{as:rct}, the influence curve
$D^*_{\psi_{\alpha_{n,v}},\, \Qbar^*_{\alpha_{n,v},n,v},\, g_0}$
is uniformly bounded for all $\alpha_{n,v} \in \Delta^{K-1}$ and
all $\Qbar^*_{\alpha_{n,v},n,v}$ in the range of the outcome.

\paragraph{Condition~(iii): Influence curve convergence}
Suppose $\alpha_{n,v} \to \alpha^*$ in
probability for some fixed $\alpha^* \in \Delta^{K-1}$---a
condition verified under the null in
Section~\ref{sec:null-theory} and under alternatives in
Section~\ref{sec:alt-theory}.
Because $\alpha \mapsto D^*_{\psi_\alpha, \Qbar_\alpha, g_0}$
is linear in $\alpha$ and
Assumption~\ref{as:nuisance} ensures
$\Qbar^*_{\alpha_{n,v},n,v} \to \Qbar_{0,\alpha^*}$ in
$L^2(\Pzero)$, we obtain
$\Pzero\{D^*_{\psi_{\alpha_{n,v}},\Qbar^*_{\alpha_{n,v},n,v},g_0}
- D^*_{\psi_{\alpha^*},\Qbar_{0,\alpha^*},g_0}\}^2 \to_p 0$,
verifying condition~(iii) with limit influence curve
$D^*(\Pzero) = D^*_{\psi_{\alpha^*}, \Qbar_{0,\alpha^*}, g_0}$.

\paragraph{Condition~(iv): Negligible remainder}
By the double-robust structure of the ATE remainder and
Assumption~\ref{as:rct}, the remainder vanishes identically:
$R_{\alpha_{n,v}}(\Qbar^*_{\alpha_{n,v},n,v}, g_0, \Qbar_0, g_0) = 0$
for all $\alpha_{n,v}$ and all $\Qbar^*_{\alpha_{n,v},n,v}$, so
condition~(iv) holds trivially.

\medskip
\noindent
Applying Theorem~\ref{thm:cvtmle-datp} yields the asymptotic
linearity of the CV-TMLE:
\begin{equation}\label{eq:al-final}
  \sqrt{n}\Bigl(
    \hat{\psi}^*_{\mathrm{CV}}
    - \frac{1}{V}\sum_v \psi_{\alpha_{n,v}}(\Pzero)
  \Bigr)
  \;\xrightarrow{d}\;
  N\bigl(0,\; \sigma^2_{\alpha^*}(\Pzero)\bigr),
\end{equation}
where
$\sigma^2_{\alpha^*}(\Pzero)
= \mathrm{Var}_{\Pzero}\!\bigl(
D^*_{\psi_{\alpha^*}, \Qbar_{0,\alpha^*}, g_0}(O)\bigr)$.

\subsection{Type~I error control under the global null}\label{sec:null-theory}

Under the one-sided global null $H_0\!: \psi_k(\Pzero) \leq 0$ for
all $k$, the signal-to-noise ratio $\theta_\alpha(\Pzero) \leq 0$ for
every $\alpha \in \Delta^{K-1}$. The oracle weight $\alpha_0$ in
\eqref{eq:oracle} is non-unique whenever two or more
$\psi_k(\Pzero) = 0$, since any convex combination on the zero
coordinates achieves the maximal value $\theta_\alpha(\Pzero) = 0$.
When exactly one $\psi_k(\Pzero) = 0$ (the rest strictly negative),
the oracle weight is uniquely the $k$th vertex $e_k$ of $\Delta^{K-1}$ and the consistency argument of
Section~\ref{sec:alt-theory} applies in principle; however, whenever
two or more $\psi_k(\Pzero) = 0$, condition~(iii) of
Theorem~\ref{thm:cvtmle-datp} cannot be verified by showing
$\hat{\alpha}(P_{n,v}) \to \alpha_0$.
The stabilization mechanism resolves this: we show that under $H_0$
the stabilized weights converge to the pre-specified reference
$\aref$, which then serves as the required fixed limit $\alpha^*$.

\begin{proposition}[Type~I error control]\label{prop:null}
Under Assumptions~\ref{as:rct}--\ref{as:nuisance} and the one-sided
global null $H_0\!: \psi_k(\Pzero) \leq 0$ for all $k$, the stabilized
CV-TMLE test has asymptotic
Type~I error rate at most $\gamma$, with equality when
$\psi_k(\Pzero) = 0$ for all $k$.
\end{proposition}

\begin{proof}

\emph{Step~1: Weight convergence.}
We show $\astab^{(v)} \xrightarrow{P} \aref$ for each fold $v$.
The $p$-value $p_{n,v}$ is computed from the asymptotic point-null
distribution of the supremum statistic (i.e., the limit distribution
when $\psi_k(\Pzero) = 0$ for all $k$).
Under any $\Pzero$ satisfying $H_0$, we show that $p_{n,v}$
asymptotically stochastically dominates a $\mathrm{Uniform}(0,1)$
distribution, in three stages.
First, for every fixed $\alpha \in \Delta^{K-1}$, define the centered statistic
\[
  T_\alpha^{(0)}
  \;=\;
  \frac{\sqrt{n_{\mathrm{train}}}\bigl(
    \hat{\psi}^*_{\alpha,\mathrm{train}} - \psi_\alpha(\Pzero)
  \bigr)}{\hat{\sigma}_{\alpha,\mathrm{train}}}.
\]
By standard TMLE asymptotics and Slutsky's theorem,
$T_\alpha^{(0)} \xrightarrow{d} N(0,1)$, and the process
$\alpha \mapsto T_\alpha^{(0)}$ converges weakly over the simplex
to the standard point-null Gaussian process driving the limit
distribution of $\sup_\alpha T_\alpha^{(0)}$.
Second, observe the exact finite-sample algebraic identity
\[
  T_\alpha
  \;=\;
  T_\alpha^{(0)}
  \;+\;
  \frac{\sqrt{n_{\mathrm{train}}}\,\psi_\alpha(\Pzero)}
       {\hat{\sigma}_{\alpha,\mathrm{train}}}.
\]
Under $H_0$, $\psi_k(\Pzero) \leq 0$ for all $k$, and since
$\alpha_k \geq 0$ this implies
$\psi_\alpha(\Pzero) = \sum_k \alpha_k \psi_k(\Pzero) \leq 0$
for every $\alpha \in \Delta^{K-1}$.
Combined with the strict positivity of the estimated standard error
$\hat{\sigma}_{\alpha,\mathrm{train}} > 0$, the shift term is
\emph{deterministically} non-positive.
Thus $T_\alpha \leq T_\alpha^{(0)}$ holds pathwise (exactly) in
every finite sample, simultaneously across all $\alpha$, and taking
the supremum preserves the inequality:
$T^*_v = \sup_\alpha T_\alpha \leq \sup_\alpha T_\alpha^{(0)}$.
Third, $p_{n,v}$ is evaluated using the asymptotic distribution of
$\sup_\alpha T_\alpha^{(0)}$, so the pathwise inequality above yields
asymptotic stochastic dominance,
$\limsup_{n \to \infty} P(p_{n,v} \leq \epsilon) \leq \epsilon$
for every $\epsilon \in (0,1)$.
In particular, $p_{n,v}$ does not vanish in probability, and because
$\tilde{C} = C\log(n) \to \infty$, the product
$\tilde{C}\, p_{n,v} \xrightarrow{P} \infty$.
After truncation at~1, the shrinkage coefficient
$\min(1, \tilde{C}\, p_{n,v}) \to 1$ in probability, and therefore
\[
  \astab^{(v)}
  = \bigl(1 - \min(1,\tilde{C}\,p_{n,v})\bigr)\,\aadapt^{(v)}
  + \min(1,\tilde{C}\,p_{n,v})\,\aref
  \;\xrightarrow{P}\; \aref.
\]

\emph{Step~2: Condition verification.}
With $\astab^{(v)} \xrightarrow{P} \aref$, condition~(iii) of
Theorem~\ref{thm:cvtmle-datp} is verified with $\alpha^* = \aref$
by exactly the argument given for condition~(iii) in
Section~\ref{sec:assumptions}.

\emph{Step~3: Size control.}
Because $\tilde{C}\, p_{n,v} \xrightarrow{P} \infty$,
it follows that $P(\tilde{C}\, p_{n,v} \geq 1) \to 1$, so the
truncation ensures that $\astab^{(v)} = \aref$ \emph{exactly}
with probability tending to~$1$.
This exact match guarantees
$\sqrt{n}(\astab^{(v)} - \aref) \xrightarrow{P} 0$, allowing us to
replace the data-adaptive estimand with the fixed reference estimand
without introducing any asymptotic bias.
Applying Theorem~\ref{thm:cvtmle-datp} with $\alpha^* = \aref$
yields the asymptotic linearity of the CV-TMLE estimator
$\hat{\psi}^*_{\mathrm{CV}}$.
We can then exactly decompose the final test statistic as
\[
  T_{\mathrm{CV}}
  \;=\;
  \frac{\sqrt{n}\bigl(\hat{\psi}^*_{\mathrm{CV}}
    - \psi_{\aref}(\Pzero)\bigr)}{\hat{\sigma}_{\mathrm{CV}}}
  \;+\;
  \frac{\sqrt{n}\,\psi_{\aref}(\Pzero)}{\hat{\sigma}_{\mathrm{CV}}}.
\]
The first term converges in distribution to $N(0,1)$ by standard
CV-TMLE asymptotics and Slutsky's theorem.
Under $H_0$, $\psi_{\aref}(\Pzero) \leq 0$, and we case-split on the
shift term.
If $\psi_{\aref}(\Pzero) = 0$, the shift term is exactly zero,
$T_{\mathrm{CV}} \xrightarrow{d} N(0,1)$, and the asymptotic
rejection probability is exactly $\gamma$.
If $\psi_{\aref}(\Pzero) < 0$, the shift term diverges to $-\infty$
in probability, driving $T_{\mathrm{CV}} \xrightarrow{P} -\infty$
and forcing the rejection probability to approach $0$.
In either case, the asymptotic Type~I error rate is at most $\gamma$,
with equality at the point null.
\end{proof}

\subsection{Validity and power under alternatives}\label{sec:alt-theory}

We now verify condition~(iii) of Theorem~\ref{thm:cvtmle-datp} under
alternatives by showing that the stabilized weights converge to a
well-defined limit, and derive the resulting validity
guarantees.  We consider fixed alternatives first, then study
asymptotic local power.

\begin{proposition}[Validity under fixed alternatives]\label{prop:alt}
Under Assumptions~\ref{as:rct}--\ref{as:nuisance} and any fixed alternative
$\Pzero$ with $\psi_k(\Pzero) > 0$ for at least one $k$, and
$\alpha_0$ the unique maximizer of $\theta_\alpha(\Pzero)$, the stabilized
CV-TMLE satisfies
\[
  \sqrt{n}\Bigl(
    \hat{\psi}^*_{\mathrm{CV}}
    - \frac{1}{V}\sum_{v=1}^V \psi_{\astab^{(v)}}(\Pzero)
  \Bigr)
  \;\xrightarrow{d}\;
  N\bigl(0,\; \sigma^2_{\alpha_0}(\Pzero)\bigr),
\]
and the cross-validated variance estimator
$\hat{\sigma}^2_{\mathrm{CV}}$ is consistent for
$\sigma^2_{\alpha_0}(\Pzero)$.
In particular, the test is consistent: its power converges to~$1$.
\end{proposition}

\begin{proof}

\emph{Step~1: Well-definedness of the oracle weight.}
Under any fixed alternative where $\psi_k(\Pzero) > 0$ for at least one
$k$, the oracle weight $\alpha_0$ in \eqref{eq:oracle} is
well-defined.
The signal-to-noise ratio
$\theta_\alpha(\Pzero) = \psi_\alpha(\Pzero)/\sigma_\alpha(\Pzero)$
is a continuous function on the compact simplex $\Delta^{K-1}$, so a maximizer exists.

\emph{Step~2: Consistency of the weight estimator.}
We show $\aadapt^{(v)} \to \alpha_0$ in probability.
Because the propensity score is known by design (Assumption~\ref{as:rct}), the training-fold
endpoint-specific TMLEs $\hat{\psi}^*_{k,\mathrm{train}}$ are
consistent for $\psi_k(\Pzero)$.
Under Assumption~\ref{as:nuisance}, the initial outcome models converge
to the true conditional means, which ensures the estimated influence
curves converge to the true efficient influence curves.
Combined with the boundedness of the influence curves
(Assumption~\ref{as:bounded}) and the law of large numbers,
the empirical covariance
$\hat{\rho}_n(k_1, k_2) \xrightarrow{P} \rho_{\Pzero}(k_1, k_2)$.
The signal-to-noise ratio $\theta_\alpha$ is a continuous function of
$(\psi_1, \ldots, \psi_K, \boldsymbol{\rho})$ for each $\alpha$, and
under the alternative $\alpha_0$ is an isolated maximizer.
By the argmax continuous mapping theorem \citep{vaart1998},
$\aadapt^{(v)} \to \alpha_0$ in probability.

\emph{Step~3: Stabilization and condition verification.}
Under the alternative, the training-fold test statistic $T^*_v$
diverges at rate $\sqrt{n}$ because
$\theta_{\alpha_0}(\Pzero) > 0$, so the associated $p$-value
$p_{n,v} \to 0$ at an exponential rate.
Since $\tilde{C} = C\log(n)$ grows only logarithmically,
$\tilde{C}\, p_{n,v} \to 0$, and therefore
$\astab^{(v)} \to \aadapt^{(v)} \to \alpha_0$ in probability.
This verifies condition~(iii) of Theorem~\ref{thm:cvtmle-datp} with
$\alpha^* = \alpha_0$.

\emph{Step~4: Asymptotic normality and consistency.}
Applying Theorem~\ref{thm:cvtmle-datp} with $\alpha^* = \alpha_0$
and the convergence $\astab^{(v)} \xrightarrow{P} \alpha_0$ from Step~3
yields the claimed asymptotic normality for the data-adaptive target
parameter.
Because $\astab^{(v)} \xrightarrow{P} \alpha_0$, the data-adaptive
estimand satisfies
$V^{-1}\sum_v \psi_{\astab^{(v)}}(\Pzero) \xrightarrow{P}
\psi_{\alpha_0}(\Pzero) > 0$
by the continuity of $\alpha \mapsto \psi_\alpha(\Pzero)$.
Variance consistency follows because the cross-validated influence
curve values $\hat{D}^*_i$ converge to
$D^*_{\psi_{\alpha_0}, \Qbar_{\alpha_0}, g_0}(O_i)$, so
$\hat{\sigma}^2_{\mathrm{CV}} \xrightarrow{P}
\sigma^2_{\alpha_0}(\Pzero)$.
Because the data-adaptive estimand converges to a strictly positive
constant and the variance converges to a finite value, the uncentered
test statistic diverges:
$T_{\mathrm{CV}} \xrightarrow{P} \infty$.
Thus, the test is consistent and power converges to~$1$.
\end{proof}

\paragraph{Local alternatives}
The consistency result above holds for every fixed alternative, but
every consistent test has power tending to~$1$ under fixed
alternatives.
To distinguish tests by their power, we consider a sequence of
\emph{local alternatives} $P_{0,n}$ converging to a fixed
distribution $\Pzero$ on the boundary of the global null
(i.e., $\psi_k(\Pzero) = 0$ for all $k$), satisfying
$\psi_k(P_{0,n}) = \varepsilon_k / \sqrt{n}$ for fixed constants
$\varepsilon_k$, at least one positive.

\begin{proposition}[Local power]\label{prop:local}
Under Assumptions~\ref{as:rct}--\ref{as:nuisance} and a sequence of
local alternatives $P_{0,n}$ with
$\psi_k(P_{0,n}) = \varepsilon_k / \sqrt{n}$ for fixed
$\varepsilon = (\varepsilon_1, \ldots, \varepsilon_K)$,
the stabilized CV-TMLE test has asymptotic local power
\[
  \pi(\varepsilon)
  \;=\;
  \bar{\Phi}\!\Bigl(
    z_{1-\gamma}
    - \frac{\textstyle\sum_k \alpha_{\mathrm{ref},k}\,\varepsilon_k}
           {\sigma_{\aref}(\Pzero)}
  \Bigr),
\]
where $\bar{\Phi} = 1 - \Phi$ is the survival function of the
standard normal distribution and $\sigma_{\aref}(\Pzero)$ denotes
the standard deviation of the influence curve at~$\aref$ under~$\Pzero$.
\end{proposition}

\begin{proof}
By standard local asymptotic normality (LAN) theory, the $n^{-1/2}$-local alternatives $\{P_{0,n}\}$ are contiguous with respect to the null \citep{vaart1998}, so any statistic that is $\Op(1)$ under the null remains $\Op(1)$ under $\{P_{0,n}\}$.
In particular, $T^*_v = \Op(1)$, and therefore $p_{n,v}$ is stochastically bounded away from zero.
Consequently,
$\tilde{C}\, p_{n,v} = C\log(n)\,p_{n,v} \xrightarrow{P} \infty$,
and the stabilized weights collapse to the reference:
$\astab^{(v)} \xrightarrow{P} \aref$.
This verifies condition~(iii) of Theorem~\ref{thm:cvtmle-datp} with
$\alpha^* = \aref$, yielding
$\sqrt{n}(\hat{\psi}^*_{\mathrm{CV}}
- \psi_{\aref}(P_{0,n})) \xrightarrow{d}
N(0, \sigma^2_{\aref}(\Pzero))$
by standard contiguity arguments.
Since
$\sqrt{n}\,\psi_{\aref}(P_{0,n})
= \sum_k \alpha_{\mathrm{ref},k}\,\varepsilon_k$,
the test statistic satisfies
$T_{\mathrm{CV}} \xrightarrow{d}
N\bigl(\sum_k \alpha_{\mathrm{ref},k}\,\varepsilon_k
/ \sigma_{\aref}(\Pzero),\; 1\bigr)$,
giving the stated rejection probability.
\end{proof}

\noindent
Proposition~\ref{prop:local} reveals that the asymptotic local power
of the stabilized CV-TMLE test equals that of the non-adaptive test
with fixed weights~$\aref$.
This is an inherent consequence of the $\log(n)$-rate stabilization:
the local signal is too weak for the training-fold test to
distinguish from the null, so the stabilization mechanism dominates
and the data-adaptive weights cannot improve upon $\aref$
asymptotically.
The practical advantage of the procedure therefore rests on
\emph{finite-sample} power gains at moderate effect sizes and on
robustness to weight misspecification---precisely the properties
investigated in the simulation studies of
Section~\ref{sec:simulations}.

\section{Simulation studies}\label{sec:simulations}

We present two simulation studies evaluating the finite-sample performance of the stabilized CV-TMLE procedure proposed in Section~\ref{sec:method}.
The first study uses a linear data-generating process (DGP) to compare the stabilized CV-TMLE with commonly used global tests and multiplicity adjustment procedures under different treatment effect patterns, isolating the effect of the testing procedure by holding the covariate-adjustment method (TMLE) fixed across all comparators.
The second study uses a realistic DGP calibrated to the ELAPRASE (idursulfase) Phase~III trial for Hunter syndrome \citep{elaprase2018}, comparing the proposed stabilized CV-TMLE against the unadjusted O'Brien rank-sum test that was used in the actual trial.
All simulations use seed 202701 and 1{,}000 replications for full reproducibility, with all tests applied at the one-sided level $\gamma = 0.025$.
The R code for both simulation studies is publicly available at \url{https://github.com/tzhou2801/CV-TMLE_global_test}.

\subsection{Study 1}\label{sec:sim-simple}

We generate $n = 50$ independent observations with $K = 2$ continuous outcomes from the following DGP:
\begin{align}
  W_1 &\sim \mathrm{Uniform}\{5, \ldots, 18\}, \quad W_2 \sim \mathrm{Bernoulli}\bigl(\mathrm{expit}(0.3)\bigr), \label{eq:sim1-covariates} \\
  A &\sim \mathrm{Bernoulli}(0.5), \label{eq:sim1-treatment} \\
  Y_k &= 1 + \beta_{A,k}\, A + \beta_{W_1,k}\, W_1 + \beta_{W_2,k}\, W_2 + \varepsilon_k, \quad \varepsilon_k \overset{\mathrm{iid}}{\sim} N(0,1), \quad k = 1, 2, \label{eq:sim1-outcome}
\end{align}
where $W_1$ and $W_2$ are baseline covariates, $A$ is the treatment indicator under 1:1 randomization, and the covariate effects are set to $(\beta_{W_1,1}, \beta_{W_1,2}) = (-0.1, -0.05)$ and $(\beta_{W_2,1}, \beta_{W_2,2}) = (0.6, 0.3)$.
We consider four treatment effect configurations to assess Type~I error control and power under varying alternatives:
S1 (global null) sets $\beta_{A,1} = \beta_{A,2} = 0$;
S2 (strong effect on $Y_1$ only) sets $\beta_{A,1} = 1.0$, $\beta_{A,2} = 0$;
S3 (equal moderate effects) sets $\beta_{A,1} = \beta_{A,2} = 0.5$;
S4 (asymmetric effects) sets $\beta_{A,1} = 0.8$, $\beta_{A,2} = 0.2$.

We compare the proposed stabilized CV-TMLE with three commonly used approaches: the O'Brien OLS global test \citep{obrien1984}, and the Holm step-down \citep{holm1979} and Hochberg step-up \citep{hochberg1988} multiplicity adjustment procedures.
The O'Brien OLS test is applied to TMLE estimates, with the influence-curve-based covariance matrix used for variance estimation.
The Holm and Hochberg procedures are applied to individual $t$-statistics constructed from the endpoint-specific TMLE estimates.
The stabilized CV-TMLE uses $V = 10$ folds, stabilization constant $C = 0.25$, and equal reference weights $\aref = (0.5, 0.5)$ corresponding to the O'Brien aggregation.
All four methods use TMLE with correctly specified outcome model and the known propensity score $g_0 = 0.5$ for covariate adjustment, so that the comparison isolates the effect of the testing procedure. The small-sample t-distribution calibration is used for all methods.
We note that this setup is somewhat idealized: in practice, the true outcome model is unknown, and it is encouraged to use the Super Learner \citep{vanderLaan2007super} to avoid model misspecification.
To illustrate the importance of covariate adjustment, we also include an unadjusted O'Brien rank-sum test for reference.

Tables~\ref{tab:sim1-power} and~\ref{tab:sim1-weights} report the Type~I error and power, and the average weights learned by the stabilized CV-TMLE, across all four scenarios.
Under the global null (S1), all methods control the Type~I error rate at the nominal level, and the stabilization mechanism shrinks the learned weights toward the reference $\aref = (0.5, 0.5)$, yielding near-equal weights on average.
When the treatment effect is concentrated on a single endpoint (S2), the stabilized CV-TMLE learns to place most of its weight on $Y_1$ and achieves the highest power.
The Holm and Hochberg procedures also perform well, as the strong marginal signal on $Y_1$ remains significant after their respective multiplicity adjustments, whereas the O'Brien OLS test, by assigning equal weight to both endpoints, substantially dilutes the concentrated signal and suffers a marked loss in power.
Conversely, when effects are equal and moderate (S3), the O'Brien OLS test is most powerful because its equal-weight aggregation is optimal under uniform effects.
Although the stabilized CV-TMLE learns correct weights on average, it incurs a power cost from finite-sample variability in the learned weights, resulting in moderately lower power than O'Brien.
The Holm and Hochberg procedures are most conservative in this scenario, as the moderate per-endpoint signals are individually insufficient to achieve significance after multiplicity correction.
Under asymmetric effects (S4), the learned weights approximately reflect the ratio of the true effect magnitudes, and the stabilized CV-TMLE achieves the highest power, outperforming all three comparators.

Finally, comparing the adjusted and unadjusted O'Brien columns in Table~\ref{tab:sim1-power} highlights the value of covariate adjustment: the unadjusted O'Brien rank-sum test suffers substantial power loss under every alternative scenario---for example, dropping from 0.657 to 0.478 in S2---underscoring the importance of covariate adjustment in rare disease trials where sample sizes are small.

\begin{center}
[Table 1 near here]
\end{center}

\begin{center}
[Table 2 near here]
\end{center}

Across all scenarios, the stabilized CV-TMLE is the only method that remains competitive regardless of the unknown treatment effect pattern: it achieves the highest power when effects are heterogeneous (S2, S4) and avoids severe power loss when effects are uniform (S3).
This robustness to the unknown effect configuration is the key practical advantage of the data-adaptive approach.

\subsection{Study 2}\label{sec:sim-realistic}

To evaluate the proposed method in a clinically realistic setting, we design a simulation calibrated to the pivotal Phase~III trial of ELAPRASE (idursulfase) for Hunter syndrome (Mucopolysaccharidosis~II, MPS~II), as described in Section~14.1 of the FDA-approved prescribing information \citep{elaprase2018}.
That trial was a 53-week, randomized, double-blind, placebo-controlled study of 96 patients (ages 5--31) with iduronate-2-sulfatase deficiency and percent-predicted forced vital capacity (\%-predicted FVC) below 80\%.
Patients were randomized to three arms: ELAPRASE 0.5~mg/kg weekly ($n = 32$), ELAPRASE 0.5~mg/kg every other week ($n = 32$), or placebo ($n = 32$).
The primary efficacy endpoint was a two-component composite score analyzed using an O'Brien rank-sum test, based on the sum of ranks of the change from baseline to Week~53 in (i) the six-minute walk test (6MWT) distance and (ii) the \%-predicted FVC.
The primary comparison (weekly versus placebo) yielded $p = 0.0049$.
Examination of the individual components showed that the weekly group experienced a 35-meter greater mean increase in 6MWT distance ($p = 0.01$) and a 4.3 percentage-point greater mean increase in \%-predicted FVC ($p = 0.07$) relative to placebo, after ANCOVA adjustment for baseline disease severity, region, and age.

Our simulation mimics the weekly-versus-placebo comparison with $n = 60$ patients (30 per arm).
Covariates include age, region, and disease severity.
Baseline 6MWT distance and \%-predicted FVC are drawn from a truncated bivariate normal distribution with covariate-dependent means and marginal standard deviations calibrated to Table~5 of the label, with FVC truncated below 80\%.
Change scores from baseline to Week~53 are generated from arm-specific bivariate normal distributions with moderate covariate effects.
We evaluate two scenarios: under the global null, the active-arm change distribution is identical to the placebo distribution; under the trial-calibrated alternative, the mean changes for both arms are set to the values observed in Table~5 of the label.
The complete data-generating process is specified in Appendix~\ref{app:dgp-study2}.

We compared the proposed stabilized CV-TMLE with the O'Brien rank-sum test used in the actual ELAPRASE trial.
The stabilized CV-TMLE is applied to the continuous change scores with covariate adjustment using a linear working model for age, region, and disease severity, $V = 10$ cross-validation folds, stabilization constant $C = 2$, and equal reference weights $\alpha_{\mathrm{ref}} = (0.5, 0.5)$.
The O'Brien rank-sum test is applied to the ranked change scores, with $5{,}000$ permutations employed for inference.

\begin{center}
[Table 3 near here]
\end{center}

Table~\ref{tab:sim2-power} reports the Type~I error and power for the two methods.
Both methods maintain the Type~I error rate at the nominal level.
Under the trial-calibrated alternative, the stabilized CV-TMLE achieves a relative power improvement of approximately 14\% over the O'Brien rank-sum test.
This gain reflects two complementary advantages of the proposed approach: covariate adjustment exploits the prognostic value of baseline age, disease severity, and region (the same covariates for which the actual trial's ANCOVA adjusted), while data-adaptive weighting allows the composite to concentrate on the more informative endpoint.
For trials where enrollment is inherently limited, this improvement translates directly into a greater chance of detecting a true treatment benefit.

\section{Conclusion}\label{sec:conclusion}

We have proposed a stabilized CV-TMLE framework for testing the global null hypothesis of no treatment benefit across multiple endpoints.
The procedure addresses two key limitations of classical global tests: it replaces fixed aggregation schemes with data-adaptive weight selection that maximizes power, and it incorporates flexible covariate adjustment through the TMLE framework.
Cross-validation separates the weight learning from estimation and inference, enabling valid inference for a data-adaptive target parameter, while the stabilization mechanism ensures Type~I error control and allows investigators to encode domain knowledge through the choice of reference weights.
Two simulation studies, one examining a range of treatment effect configurations and the other calibrated to a Phase~III rare disease trial, demonstrate valid Type~I error control and power gains over commonly used global tests and multiplicity adjustment methods.
Because the treatment effect pattern across endpoints is typically unknown at the design stage, the ability to learn the optimal composite weighting from the observed data, rather than committing to a fixed scheme based on anticipated effect sizes, is of considerable practical value.

Several directions for future work remain.
First, the stabilization mechanism introduced here shrinks the data-adaptive weights toward a fixed reference based on the training-fold $p$-value; investigating alternative shrinkage strategies---such as empirical Bayes shrinkage---may yield improved finite-sample properties.
Second, missing data and informative dropout are common in rare disease trials; extending the framework to accommodate these is an important practical direction.
Third, although our development and simulations are motivated by rare disease trials with small sample sizes, the stabilized CV-TMLE framework is broadly applicable to any setting that involves testing a global null across multiple endpoints, including large confirmatory trials, platform trials, and observational comparative effectiveness studies.

\section*{Disclosure of interest}
Tianyue Zhou reports tuition and stipend support from a philanthropic gift from the Novo Nordisk corporation to the University of California, Berkeley to support the Joint Initiative for Causal Inference.

The views are those of the author(s) and do not necessarily represent the official views of, nor an endorsement by, FDA, HHS, or the U.S. Government.



\appendix

\section{Data-generating process for Study~2}\label{app:dgp-study2}

This appendix specifies the complete data-generating process for the realistic simulation in Section~\ref{sec:sim-realistic}.
All numerical parameters are calibrated to Table~5 of the ELAPRASE prescribing information \citep{elaprase2018}.

\paragraph{Covariates}
For each of $n = 60$ subjects:
\begin{align}
  \text{Age}_i &\sim \mathcal{N}(15,\; 6^2) \text{ truncated to } [5, 31], \notag \\
  \text{Region}_i &\in \{\text{NA},\; \text{EU},\; \text{Other}\} \text{ with probabilities } (0.40,\; 0.40,\; 0.20), \notag \\
  \text{Severity}_i &\in \{\text{mild},\; \text{moderate},\; \text{severe}\} \text{ with probabilities } (0.31,\; 0.38,\; 0.31).
\end{align}

\paragraph{Centered design matrix}
Define the five-column centered design matrix
\[
  X_i = \bigl(\text{age10}_i,\; \text{sev\_mod}_i^c,\; \text{sev\_sev}_i^c,\; \text{EU}_i^c,\; \text{Other}_i^c\bigr),
\]
where
\begin{align}
  \text{age10}_i &= \bigl(\text{Age}_i - \mathbb{E}[\text{Age}]\bigr) / 10, \notag \\
  \text{sev\_mod}_i^c &= \mathbf{1}(\text{Severity}_i = \text{moderate}) - 0.38, \notag \\
  \text{sev\_sev}_i^c &= \mathbf{1}(\text{Severity}_i = \text{severe}) - 0.31, \notag \\
  \text{EU}_i^c &= \mathbf{1}(\text{Region}_i = \text{EU}) - 0.40, \notag \\
  \text{Other}_i^c &= \mathbf{1}(\text{Region}_i = \text{Other}) - 0.20,
\end{align}
where $\mathbb{E}[\text{Age}]$ is the population mean of the truncated normal distribution and age is rescaled by a factor of 10 to place it on a comparable scale to the indicator covariates.

\paragraph{Covariate effect matrices}
Linear covariate effects on baseline and change outcomes are specified by the $2 \times 5$ matrices
\[
  M_B = \begin{pmatrix} -15 & -60 & -120 & -10 & 5 \\ -1.5 & -6 & -12 & -1.5 & 0.75 \end{pmatrix}, \qquad
  M_D = \begin{pmatrix} -8 & -18 & -40 & -4 & 2 \\ -1 & -5 & -10 & -0.8 & 0.5 \end{pmatrix},
\]
where each row corresponds to an outcome (6MWT, FVC) and each column to a covariate in $X_i$.
The coefficients are calibrated so that the covariate-driven component explains approximately 12--18\% of the total outcome variance.

\paragraph{Baseline outcomes}
The baseline outcomes $B_i = (B_{\text{6MWT},i},\; B_{\text{FVC},i})^\top$ are generated as
\[
  B_i \mid X_i \sim \text{TruncatedBVN}\!\left(\mu_B + M_B X_i,\; \Sigma_{B,\text{resid}};\; [50, 650] \times [20, 80)\right),
\]
with target marginal parameters $\mu_B = (392.5,\; 55.45)$, $\text{SD}_B = (107,\; 14)$, and $\rho_B = 0.30$.
The residual covariance $\Sigma_{B,\text{resid}} = \Sigma_{B,\text{target}} - M_B \,\text{Var}(X) \, M_B^\top$ is computed so that the marginal covariance approximately matches $\Sigma_{B,\text{target}}$.

\paragraph{Treatment assignment}
Treatment is assigned by balanced randomization: $A_i \in \{0, 1\}$ with exactly $n/2$ subjects per arm.

\paragraph{Change scores}
The change from baseline to Week~53, $\Delta_i = (\Delta_{\text{6MWT},i},\; \Delta_{\text{FVC},i})^\top$, is generated as
\[
  \Delta_i \mid A_i, X_i \sim \text{BVN}\!\left(\mu_{\Delta}^{(A_i)} + M_D X_i,\; \Sigma_{\Delta,\text{resid}}^{(A_i)}\right),
\]
with the following arm-specific parameters:

\smallskip
\begin{center}
\begin{tabular}{lccc}
  \toprule
  & Mean change & SD & Correlation \\
  \midrule
  Placebo ($A = 0$) & $(7,\; 0.8)$ & $(54,\; 9.6)$ & $\rho_D = 0.25$ \\
  Active ($A = 1$, trial-calibrated) & $(44,\; 3.4)$ & $(70,\; 10.0)$ & $\rho_D = 0.25$ \\
  \bottomrule
\end{tabular}
\end{center}

\smallskip
\noindent
Under the global null scenario, the active-arm parameters are set equal to the placebo parameters (including equal variances).
The residual covariance matrices $\Sigma_{\Delta,\text{resid}}^{(a)}$ are computed analogously to the baseline case.

\paragraph{Analysis outcomes}
The analysis endpoints are the continuous changes $(\Delta\text{6MWT}_i,\; \Delta\text{FVC}_i)$, used by the stabilized CV-TMLE, and the rank scores $Y_{k,i} = \text{rank}(\Delta_{k,i})$ across all subjects, used by the O'Brien rank-sum test.



\clearpage
\begin{table}[p]
  \centering
  \caption{Type~I error and power under the Study~1 DGP ($n = 50$, $\gamma = 0.025$, 1{,}000 replications). The first four methods use TMLE-based covariate adjustment; the last column shows the unadjusted O'Brien rank-sum test.}
  \label{tab:sim1-power}
  \begin{tabular}{llccccc}
    \toprule
    Scenario & $(\beta_{A,1}, \beta_{A,2})$ & Holm & Hochberg & O'Brien & Stab.\ CV-TMLE & Unadj.\ O'Brien \\
    \midrule
    S1 (Type~I error) & $(0, 0)$       & 0.021 & 0.021 & 0.022 & 0.023 & 0.028 \\
    S2 (Strong $Y_1$) & $(1.0, 0)$     & 0.858 & 0.858 & 0.657 & 0.873 & 0.478 \\
    S3 (Equal)         & $(0.5, 0.5)$   & 0.472 & 0.487 & 0.645 & 0.562 & 0.526 \\
    S4 (Asymmetric)    & $(0.8, 0.2)$   & 0.668 & 0.672 & 0.648 & 0.700 & 0.504 \\
    \bottomrule
  \end{tabular}
  \par\smallskip
  {\footnotesize \textit{Note.} TMLE = targeted maximum likelihood estimation; CV-TMLE = cross-validated TMLE; Stab.\ CV-TMLE = stabilized CV-TMLE; Unadj.\ = unadjusted analysis; Holm = Holm step-down procedure; Hochberg = Hochberg step-up procedure; O'Brien = O'Brien (1984) combined-test statistic; DGP = data-generating process; $n$ = total sample size; $\gamma$ = one-sided significance level; $C$ = stabilization constant ($C = 0.25$ in this study); $(\beta_{A,1}, \beta_{A,2})$ are the treatment effects on endpoints $Y_1$ and $Y_2$ in the linear outcome model. Entries are empirical rejection rates over 1{,}000 replications. Type~I error rate is reported under the global null.}
\end{table}

\clearpage
\begin{table}[p]
  \centering
  \caption{Average stabilized CV-TMLE composite weights across 1{,}000 replications under the Study~1 DGP.}
  \label{tab:sim1-weights}
  \begin{tabular}{lcc}
    \toprule
    Scenario & $\bar{\alpha}_1$ & $\bar{\alpha}_2$ \\
    \midrule
    S1 (Global null)   & 0.510 & 0.490 \\
    S2 (Strong $Y_1$)  & 0.887 & 0.113 \\
    S3 (Equal)         & 0.496 & 0.504 \\
    S4 (Asymmetric)    & 0.764 & 0.236 \\
    \bottomrule
  \end{tabular}
  \par\smallskip
  {\footnotesize \textit{Note.} Entries $\bar{\alpha}_1$ and $\bar{\alpha}_2$ are the average data-adaptive weights placed on endpoints $Y_1$ and $Y_2$ by the stabilized CV-TMLE, computed across 1{,}000 replications. S1 = global null, $(\beta_{A,1}, \beta_{A,2}) = (0, 0)$; S2 = strong effect on $Y_1$, $(1.0, 0)$; S3 = equal moderate effects, $(0.5, 0.5)$; S4 = asymmetric effects, $(0.8, 0.2)$. Stab.\ CV-TMLE = stabilized cross-validated targeted maximum likelihood estimation; DGP = data-generating process.}
\end{table}

\clearpage
\begin{table}[p]
  \centering
  \caption{Type~I error and power for the unadjusted O'Brien rank-sum test and the stabilized CV-TMLE under the realistic DGP calibrated to the ELAPRASE trial ($n = 60$, $\gamma = 0.025$, 1{,}000 replications).}
  \label{tab:sim2-power}
  \begin{tabular}{llcc}
    \toprule
    Scenario & Metric & O'Brien (unadj.) & Stab.\ CV-TMLE \\
    \midrule
    Global null                  & Type~I error rate & 0.027 & 0.024 \\
    Trial-calibrated alternative & Power             & 0.529 & 0.602 \\
    \bottomrule
  \end{tabular}
  \par\smallskip
  {\footnotesize \textit{Note.} CV-TMLE = cross-validated targeted maximum likelihood estimation; Stab.\ CV-TMLE = stabilized CV-TMLE; Unadj.\ = unadjusted analysis; O'Brien = O'Brien (1984) combined-test statistic; DGP = data-generating process calibrated to the Phase~III ELAPRASE rare-disease trial; ELAPRASE = idursulfase; ``Trial-calibrated alternative'' refers to active-arm mean changes calibrated to the values reported in the ELAPRASE prescribing information (Section~14.1); $n$ = total sample size ($n = 60$, with 30 per arm); $\gamma$ = one-sided significance level; $C$ = stabilization constant ($C = 2$ in this study). Entries are empirical rejection rates over 1{,}000 replications.}
\end{table}

\end{document}